\documentclass[11pt]{article}

\usepackage{mathpazo}
\usepackage{amsfonts}
\usepackage{amsmath}
\usepackage{amssymb}
\usepackage{latexsym}
\usepackage{relsize}
\usepackage{mathtools}
\usepackage{sectsty}
\usepackage{textcomp}
\usepackage[usenames,dvipsnames,svgnames,table]{xcolor}
\usepackage{appendix}
\usepackage{listings}
\usepackage{cancel}
\usepackage{booktabs}
\usepackage{siunitx}

\usepackage[margin=1in]{geometry}
\usepackage{hyperref}
\hypersetup{pdfpagemode=UseNone}

\usepackage{amsthm}
\newtheorem{theorem}{Theorem}[section]

\newtheorem{lemma}{Lemma}[section]
\newtheorem{prop}{Proposition}[section]

\theoremstyle{definition}

\newtheorem{example}{Example}[section]

\newcommand{\tinyspace}{\mspace{1mu}}
\newcommand{\microspace}{\mspace{0.5mu}}
\newcommand{\op}[1]{\operatorname{#1}}
\newcommand{\tr}{\operatorname{Tr}}

\renewcommand{\int}{\operatorname{int}}

\newcommand{\reg}[1]{\mathsf{#1}}

\newcommand{\abs}[1]{\left\lvert #1 \right\rvert}
\newcommand{\bigabs}[1]{\bigl\lvert #1 \bigr\rvert}
\newcommand{\Bigabs}[1]{\Bigl\lvert #1 \Bigr\rvert}

\newcommand{\ip}[2]{\langle #1 , #2\rangle}
\newcommand{\bigip}[2]{\bigl\langle #1, #2 \bigr\rangle}
\newcommand{\Bigip}[2]{\Bigl\langle #1, #2 \Bigr\rangle}

\newcommand{\norm}[1]{\left\lVert\tinyspace #1 \tinyspace\right\rVert}
\newcommand{\bignorm}[1]{\bigl\lVert\tinyspace #1 \tinyspace\bigr\rVert}
\newcommand{\Bignorm}[1]{\Bigl\lVert\tinyspace #1 \tinyspace\Bigr\rVert}
\newcommand{\biggnorm}[1]{\biggl\lVert\tinyspace #1 \tinyspace\biggr\rVert}
\newcommand{\Biggnorm}[1]{\Biggl\lVert\tinyspace #1 \tinyspace\Biggr\rVert}

\newcommand{\ket}[1]{
  \lvert\microspace #1 \microspace \rangle}

\newcommand{\bra}[1]{
  \langle\microspace #1 \microspace \rvert}

\def\I{\mathbb{1}}

\def\BB84{\mathsf{BB84}}

\newcommand{\setft}[1]{\mathrm{#1}}
\newcommand{\Density}{\setft{D}}
\newcommand{\Pos}{\setft{Pos}}

\newcommand{\Herm}{\setft{Herm}}
\newcommand{\Lin}{\setft{L}}

\makeatletter
\providecommand*{\cupdot}{%
  \mathbin{%
    \mathpalette\@cupdot{}%
  }%
}
\newcommand*{\@cupdot}[2]{%
  \ooalign{%
    $\m@th#1\cup$\cr
    \hidewidth$\m@th#1\cdot$\hidewidth
  }%
}
\makeatother

\def\complex{\mathbb{C}}

\newenvironment{mylist}[1]{\begin{list}{}{
	\setlength{\leftmargin}{#1}
	\setlength{\rightmargin}{0mm}
	\setlength{\labelsep}{2mm}
	\setlength{\labelwidth}{8mm}
	\setlength{\itemsep}{0mm}}}
	{\end{list}}

\usepackage[affil-it]{authblk}

\newcommand\A{\mathcal{A}}
\newcommand\B{\mathcal{B}}

\newcommand\R{\mathcal{R}}

\renewcommand\S{\mathcal{S}}

\renewcommand\H{\mathcal{H}}

\begin{document}

\title{\Large\bf Extended nonlocal games and monogamy-of-entanglement games}

\author[1,2]{Nathaniel Johnston}
\author[3]{Rajat Mittal}
\author[4]{Vincent Russo}
\author[4,5]{John Watrous}

\affil[1]{\small Department of Mathematics and Computer Science, Mount 
  Allison University}
\affil[2]{\small Institute for Quantum Computing and Department of 
  Combinatorics and Optimization, University of Waterloo}
\affil[3]{\small Department of Computer Science and Engineering, IIT Kanpur}
\affil[4]{\small Institute for Quantum Computing and School of Computer 
  Science, University of Waterloo}
\affil[5]{\small Canadian Institute for Advanced Research, Toronto}

\date{October 28, 2015}
\maketitle 

\begin{abstract}
  We study a generalization of nonlocal games---which we call
  \emph{extended nonlocal games}---in which the players, Alice and Bob,
  initially share a tripartite quantum state with the referee.
  In such games, the winning conditions for Alice and Bob may depend on
  outcomes of measurements made by the referee, on its part of the shared
  quantum state, in addition to Alice and Bob's answers to randomly selected
  questions.
  Our study of this class of games was inspired by the 
  \emph{monogamy-of-entanglement} games introduced by Tomamichel, Fehr,
  Kaniewski, and Wehner, which they also generalize.
  We prove that a natural extension of the Navascu\'{e}s--Pironio--Ac\'{i}n
  hierarchy of semidefinite programs converges to the optimal commuting
  measurement value of extended nonlocal games, and we prove two extensions of
  results of Tomamichel et al.~concerning monogamy-of-entanglement games.
\end{abstract}


\section{Introduction} \label{sec:introduction}

\subsubsection*{Nonlocal games}

The nonlocal games model---although not always so named or defined
explicitly---has been studied in theoretical physics and classical complexity
theory for many years.
In theoretical physics, nonlocal games provide a natural framework in which
Bell inequality experiments, proposed by Bell \cite{Bell1964} in 1964 and
subsequently studied by Clauser, Horne, Shimoney, and Holt \cite{Clauser1969}
and many others, may be framed.
In classical complexity theory, nonlocal games provide a simple, abstract
model through which two-prover (or general multi-prover) interactive proof
systems have often been analyzed \cite{Ben-Or1988,Fortnow1989,Babai1991,Feige1991,Feige1994,Raz1998}.
These two independent lines of research were merged in the context of quantum
information and computation, and the result has been an active topic of
research \cite{Cleve2004, Brassard2005, Cleve2008, Doherty2008, Kempe2010,Kempe2010a,Kempe2011,Junge2011a,BuhrmanFS14,Regev2013,Dinur2013,Vidick2013,CM14}.

Mathematically speaking, a \emph{nonlocal game} is a cooperative game of
incomplete information played by two players, conventionally named \emph{Alice}
and \emph{Bob}.
The game is run by a \emph{referee}, who begins the game by selecting a pair of
questions $(x,y)$ at random according to a fixed probability distribution, and
then sends $x$ to Alice and $y$ to Bob.
Communication between Alice and Bob is forbidden during the game---without
knowing the other player's question (or answer), Alice and Bob must respond
with answers $a$ and $b$, respectively.
Upon receiving these answers, the referee evaluates a predicate $V(a,b | x,y)$
that determines whether Alice and Bob win or lose the game.
(More generally, the function $V$ may take arbitrary real values that represent
pay-offs for Alice and Bob.)
It is assumed that Alice and Bob have complete knowledge of the function $V$
and of the probability distribution from which the question pairs are drawn,
and are free to agree before the game starts on a joint strategy.

Different classes of strategies for nonlocal games may be considered.
For instance, Alice and Bob may use a \emph{classical strategy} in which
they answer \emph{deterministically}, with $a$ and $b$ determined by functions
of $x$ and $y$, respectively, or they may make use of \emph{randomness}
(which happens not to offer any advantages over an optimally chosen
deterministic strategy when their goal is to maximize their winning
probability or expected pay-off).
Alternatively, one may consider \emph{quantum strategies} for Alice and Bob,
where they initially share a joint quantum system, and allow their answers
$a$ and $b$ to be determined by the outcomes of measurements on this shared
system.
Within this category of strategies, one may consider different
sub-classifications, including strategies in which the size of the shared
quantum state available to Alice and Bob is limited, or strategies in which the
more conventional bipartite tensor product structure of a quantum system shared
between two individuals is relaxed to the requirement that Alice and Bob make
use of \emph{commuting measurements} on a single Hilbert space.

For each type of strategy, one may speak of the \emph{value} of a given
nonlocal game with respect to that strategy type, which is the supremum value
of the probability for Alice and Bob to win (or the supremum value of Alice and
Bob's expected pay-off) over all strategies of the given type.

\subsubsection*{Extended nonlocal games}

In this paper, we consider a generalization of nonlocal games in which the
\emph{referee also holds a quantum system}, provided to it by Alice and Bob at
the start of the game.
  The game begins in a similar way to a nonlocal game, with the referee
  selecting a pair of questions $(x,y)$ from the Cartesian product $X\times Y$
  of two alphabets $X$ and $Y$, according to a fixed probability distribution.
  The question $x$ is sent to Alice, who must respond with an answer $a$ from
  a fixed alphabet $A$, and likewise $y$ is sent to Bob, who must respond with
  a symbol $b$ from a fixed alphabet $B$.
Now, however, the outcome of the game is not directly determined as the value
$V(a,b|x,y)$ of a predicate or real-valued pay-off function $V$, but rather by
the result of a measurement performed by the referee on its share of the state
initially provided to it by Alice and Bob.
We will assume, more specifically, that Alice and Bob's pay-off is determined
by an \emph{observable} $V(a,b | x,y) \in \Herm(\complex^m)$, where $m$ denotes
the dimension of the referee's quantum system---so if Alice and Bob's response
$(a,b)$ to the question pair $(x,y)$ leaves the referee's system in the quantum
state
\begin{equation}
  \rho^{x,y}_{a,b} \in \Density(\complex^m),
\end{equation}
then their pay-off will be the real-number value
\begin{equation}
  \label{eq:pay-off}
  \Bigip{V(a,b|x,y)}{\rho^{x,y}_{a,b}}
\end{equation}
(where $\ip{M}{N} = \tr(M^{\ast} N)$ is the standard Hilbert--Schmidt inner
product on $m\times m$ matrices).
If one wishes to consider that the referee makes a binary-valued decision,
representing that Alice and Bob either win or lose the game, then it may be
required that each $V(a,b|x,y)$ is a measurement operator corresponding to the
winning outcome, so that \eqref{eq:pay-off} represents the probability that
Alice and Bob win conditioned on $(x,y)$ having been answered with $(a,b)$.
It is evident that games of this form, which we call
\emph{extended nonlocal games}, include ordinary nonlocal games as a special
case;
ordinary nonlocal games may be expressed as extended nonlocal games
for which $m=1$, meaning that the referee's quantum system is a trivial,
one-dimensional system. 

Similar to an ordinary nonlocal game, one may consider a variety of possible
strategies for Alice and Bob in an extended nonlocal game.
In particular, there are classes of strategies that are analogous to
classical strategies, standard quantum strategies, and commuting measurement
strategies.
Further details on these different classes of strategies can be found
later in Section~\ref{sec:extended-nonlocal-games}.

  The general notion of extended nonlocal games was previously considered
  by Fritz \cite{Fritz2012}.
  In particular, Fritz considered a class of games, called
  \emph{bipartite steering games}, which are essentially extended nonlocal
  games in which the referee randomly chooses to ask either Alice or Bob a
  question.
  Extended nonlocal games may also be viewed as being equivalent to
  multipartite steering inequalities, in a similar way to the equivalence
  between nonlocal games and Bell inequalities.
  Multipartite steering inequalities and related notions were studied in
  the recent papers \cite{Cavalcanti2015} and \cite{Sainz2015}.

\subsubsection*{Monogamy-of-entanglement games}

Extended nonlocal games also generalize
\emph{monogamy-of-entanglement games}, which were introduced by Tomamichel,
Fehr, Kaniewski, and Wehner \cite{Tomamichel2013}.
Monogamy-of-entanglement games, which also have relevance to the problem of
position-based cryptography, provide a framework to conceptualize the
fundamental monogamy property exhibited by entangled
qubits~\cite{Coffman2000}.
In short, this property states that for three possibly entangled qubits 
$\reg{X}$, $\reg{Y}$, and $\reg{Z}$, that if $\reg{X}$ and $\reg{Y}$
are maximally entangled, then $\reg{Z}$ must be completely uncorrelated with
$\reg{X}$ and $\reg{Y}$, and likewise for any permutation of these three
qubits.
This phenomena has been studied in a number of other
works~\cite{Terhal2001, Terhal2004, Koashi2004, Osborne2006}. 

A monogamy-of-entanglement game is a game played in a similar way to an
extended nonlocal game, as described above.
Specifically, Alice and Bob initially supply the referee with a quantum system,
the referee selects a single question $x\in X$ at random, sends this question 
to both Alice and Bob, performs a measurement
\begin{equation}
  \bigl\{ \Pi^x_a\,:\,a\in A\bigr\}
\end{equation}
on its quantum system, and declares Alice and Bob winners if and only if they
both respond with the same outcome $a\in A$ that the referee's measurement
produced.
Such a game is represented as an extended nonlocal game by taking $Y = X$ and
$B = A$ and setting $V(a,a|x,x) = \Pi^x_a$ for each choice of $x\in X$ and
$a\in A$, as well as $V(a,b|x,x) = 0$ for $a\not=b$.
In addition one may define $V(a,b|x,y)$ arbitrarily for all
$x\not=y$ and all $a,b\in A$;
these matrices are irrelevant to the description of the game because the
referee never asks a question pair $(x,y)$ where $x\not=y$ in a
monogamy-of-entanglement game.

\subsubsection*{Motivation and summary of results}

By studying extended nonlocal games we hope to identify commonalities between
nonlocal games and monogamy-of-entanglement games, and to potentially gain
insights on both models through this type of generalization.
We prove the following results:

\begin{mylist}{\parindent}
\item[1.]
  An extension of the NPA hierarchy of semidefinite programs to extended
  nonlocal games.
\end{mylist}

\noindent
Navascu\'es, Pironio, and Ac\'in \cite{Navascues2007,Navascues2008} proved that
the \emph{commuting measurement value} of a nonlocal game can be expressed
through a sequence of semidefinite programs.
The optimum values of the semidefinite programs in this sequence are
nonincreasing, each establishes an upper-bound on the value of the given game,
and the sequence of optimum values necessarily converges to the true commuting
measurement value of the game.
By extending this method, we describe a sequence of semidefinite programs,
for a given extended nonlocal game, that upper-bounds and converges to the
commuting measurement value of the extended nonlocal game in a similar way.
We note that a semidefinite program corresponding to an
  intermediate level of the hierarchy we define, between the first and second
  level, was considered in~\cite{Sainz2015}.

\begin{mylist}{\parindent}
\item[2.]
  Results on monogamy-of-entanglement games with two questions.
\end{mylist}

\noindent
We prove two facts about monogamy-of-entanglement games for the case in which
the question set $X$ contains just two elements that extend results of
Tomamichel, Fehr, Kaniewski, and Wehner.
First, we prove that Alice and Bob can always achieve the quantum value of such
a game by using a strategy that does not require them to store quantum
information: they provide the referee with a chosen state at the start of the
game, but act classically thereafter.
We also provide an example of a monogamy-of-entanglement game, in which the
question set $X$ has 4 elements and the answer set $A$ has 3 elements, for
which Alice and Bob must store quantum information to play optimally, implying
that this result on two-question monogamy-of-entanglement games does not
generalize to larger question sets.
Second, we prove that a bound of Tomamichel, Fehr, Kaniewski, and Wehner
concerning parallel repetition of monogamy-of-entanglement games defined by
projective measurements is tight for two-question games, implying a strong
parallel repetition property for such games.


\subsubsection*{Organization of the paper}

In Section~\ref{sec:extended-nonlocal-games} we formally define the extended 
nonlocal game model, and consider the types of strategies that Alice and Bob
may use. 
In Section~\ref{sec:upper-bound-the-npa-referee-hierarchy}, we present
an extension of the NPA hierarchy, and prove that it converges to the commuting measurement value
  of a given extended nonlocal game.
In Section~\ref{sec:monogamy-of-entanglement-games}, we consider the class of
monogamy-of-entanglement-games and prove some results on
monogamy-of-entanglement-games with two questions.


\section{Extended nonlocal games} \label{sec:extended-nonlocal-games}

As was summarized in the introduction, an \emph{extended-nonlocal game} $G$ is
defined by a pair $(\pi,V)$, where $\pi$ is a probability distribution of the
form
\begin{equation}
  \pi: X \times Y \rightarrow [0,1]
\end{equation}
on the Cartesian product of two finite and nonempty sets $X$ and $Y$, and $V$
is a function of the form
\begin{equation}
  V: A \times B \times X\times Y \rightarrow \Herm(\complex^m),
\end{equation}
for $X$ and $Y$ as above, $A$ and $B$ being finite and nonempty sets, and
$m$ being a positive integer.
The sets $X$ and $Y$ represent sets of questions asked to Alice and Bob,
$A$ and $B$ represent Alice and Bob's sets of answers, and
$m$ represents the size of the quantum system initially provided to the referee
by Alice and Bob.
We write $\R = \complex^m$ to denote its corresponding complex vector space
for convenience.

As a result of Alice and Bob responding to the question pair $(x,y)$ with the
answer pair $(a,b)$, the referee's quantum system will be left in a quantum
state
\begin{equation}
  \rho^{x,y}_{a,b} \in \Density(\R),
\end{equation}
which is an $m\times m$ density operator.
The pay-off for Alice and Bob in this situation is given by the real number
\begin{equation}
  \Bigip{V(a,b|x,y)}{\rho^{x,y}_{a,b}}.
\end{equation}

\subsubsection*{Standard quantum strategies}

As suggested in the introduction, there are multiple classes of strategies that
may be considered for extended nonlocal games.
We will begin with \emph{standard quantum strategies}, which represent what is
arguably the most natural form of quantum strategy for the players Alice and
Bob in an extended nonlocal game.
A strategy of this form consists of finite-dimensional complex Hilbert spaces
$\A$ and $\B$ for Alice and Bob, respectively, a quantum state
$\rho\in\Density(\R\otimes\A\otimes\B)$, and two collections of measurements,
\begin{equation}
  \bigl\{A^x_a\,:\,a\in A\bigr\} \subset \Pos(\A)
  \quad\text{and}\quad
  \bigl\{B^y_b\,:\,b\in B\bigr\} \subset \Pos(\B),
\end{equation}
for each $x\in X$ and $y\in Y$, respectively.
That is, one has that
\begin{equation}
  \sum_{a\in A} A^x_a = \I_{\A}
  \quad\text{and}\quad
  \sum_{b\in B} B^y_b = \I_{\B}
\end{equation}
for each $x\in X$ and $y\in Y$.

When the game is played, Alice and Bob present the referee with a quantum
system so that the three parties share the state $\rho$.
The referee chooses $(x,y) \in X\times Y$ at random, according to the
probability distribution $\pi$, and sends $x$ to Alice and $y$ to Bob.
Alice measures her portion of $\rho$ with respect to the measurement
$\{A_a^x : a \in A\}$, and sends the result $a\in A$ of this measurement to the
referee.
Bob does likewise, sending the outcome $b\in B$ of the measurement
$\{B_b^y : b \in B\}$ to the referee.
Finally, the referee measures its quantum system and assigns a pay-off, as
specified by the observable $V(a,b|x,y)$.
The expected pay-off for such a strategy in the game $G = (\pi,V)$ is given by
\begin{equation}
  \sum_{(x,y)\in X\times Y} \pi(x,y)
  \sum_{(a,b)\in A\times B}
  \Bigip{V(a,b|x,y) \otimes A^x_a \otimes B^y_b}{\rho}.
\end{equation}
It is a simple consequence of Naimark's theorem that any strategy for Alice and
Bob that makes use of non-projective measurements can be simulated by
a projective measurement strategy, so there is no loss of generality in
restricting one's attention to projective measurements
$\{A_a^x : a \in A\}$ and $\{B_b^y : b \in B\}$ for Alice and Bob.

When analyzing a strategy for Alice and Bob as described above, it is
convenient to define a function 
$K:A\times B\times X\times Y \rightarrow \Pos(\R)$ as 
\begin{align}\label{eq:partial-trace-alice-bob}
  K(a,b|x,y) = \tr_{\A \otimes \B} 
  \Bigl( \bigl(\I_{\R} \otimes A_a^x \otimes B_b^y\bigr) \rho \Bigr)
\end{align}
for each $x\in X$, $y \in Y$, $a\in A$, and $b\in B$.

  We will refer to the function $K$ as an \emph{assemblage}, as it is
  representative of the notion of an assemblage in the context of tripartite
  quantum steering~\cite{Cavalcanti2015,Sainz2015}.
The operators output by this function represent the \emph{unnormalized} states
of the referee's quantum system when Alice and Bob respond to the question pair
$(x,y)$ with the answer pair $(a,b)$.
In particular, one has that $\tr(K(a,b|x,y))$ is the probability with which
Alice and Bob answer $(a,b)$ for the question pair $(x,y)$, and normalizing
this operator (assuming it is nonzero) yields the state 
\begin{equation}
  \rho^{x,y}_{a,b} =
  \frac{K(a,b|x,y)}{\tr(K(a,b|x,y))}
\end{equation}
of the referee's system conditioned on this question and answer pair.
Note that the function $K$ completely determines the performance of Alice and
Bob's strategy for $G$.
In particular, Alice and Bob's expected pay-off is represented as
\begin{equation}
  \label{eq:expected-pay-off}
  \sum_{(x,y) \in X\times Y} \pi(x,y) \sum_{(a,b) \in A \times B}
  \Bigip{V(a,b|x,y)}{K(a,b|x,y)}.
\end{equation}

For a given extended nonlocal game $G = (\pi,V)$, we write $\omega^{\ast}(G)$
to denote the \emph{quantum value} of $G$, which is the supremum value of
Alice and Bob's expected pay-off over all standard quantum strategies for $G$.

\subsubsection*{Unentangled strategies}

Next we consider a much more restricted form of strategy called an
\emph{unentangled strategy}.
These are standard quantum strategies for which the state
$\rho\in\Density(\R\otimes\A\otimes\B)$ initially prepared by Alice and Bob
is fully separable, meaning that it takes the form
\begin{equation}
  \rho = \sum_{j = 1}^N p_j\, \rho^{\reg{R}}_j\otimes
  \rho^{\reg{A}}_j\otimes \rho^{\reg{B}}_j
\end{equation}
for a probability vector $(p_1,\ldots,p_N)$ and density operators
\begin{equation}
  \rho^{\reg{R}}_1,\ldots,\rho^{\reg{R}}_N\in\Density(\R),\quad
  \rho^{\reg{A}}_1,\ldots,\rho^{\reg{A}}_N\in\Density(\A),\quad\text{and}\quad
  \rho^{\reg{B}}_1,\ldots,\rho^{\reg{B}}_N\in\Density(\B).
\end{equation}

One may prove that any unentangled strategy is equivalent to one in which Alice
and Bob store only classical information once the referee's quantum system has
been provided to it.
Indeed, any such strategy is equivalent to one given by a convex combination of
\emph{deterministic strategies}, in which Alice and Bob initially provide the
referee with a fixed pure state $\rho = u u^{\ast}\in\Density(\R)$, and respond
to questions deterministically, with Alice responding to $x\in X$ with 
$a = f(x)$ and Bob responding to $y\in Y$ with $b = g(y)$ for functions
$f:X\rightarrow A$ and $g:Y\rightarrow B$.

For a given game $G = (\pi,V)$, we write $\omega(G)$ to denote the
\emph{unentangled value} of $G$, which is the supremum value for Alice and
Bob's expected pay-off in $G$ over all unentangled strategies.
It follows by convexity that this supremum value is necessarily achieved by
some deterministic strategy, and can be represented as
\begin{equation}
  \omega(G) = \max_{f,g} \Biggnorm{\sum_{(x,y) \in X\times Y}
    \pi(x,y) V\bigl(f(x),g(y) | x,y\bigr)}
\end{equation}
where the maximum is over all functions $f:X \rightarrow A$ and
$g:Y\rightarrow B$.

\subsubsection*{commuting measurement strategies}

The last type of strategy we consider for Alice and Bob in an extended nonlocal
game is a \emph{commuting measurement strategy}, which is a (potentially) more
general type of strategy than a standard quantum strategy.
A commuting measurement strategy is similar to a standard quantum strategy,
except now the bipartite tensor product space $\A\otimes\B$ shared by Alice and
Bob is replaced by a single (possibly infinite-dimensional) Hilbert space $\H$.
Alice and Bob initially prepare a quantum state
\begin{equation}
  \rho \in \Density(\R\otimes\H)
\end{equation}
and give to the referee its portion of this state.
Alice and Bob's measurements on $\A$ and $\B$ are replaced by measurements on
the space $\H$, so that
\begin{equation}
  \bigl\{ A^x_a\,:\,x\in X,\: a\in A\bigr\}
  \quad\text{and}\quad
  \bigl\{ B^y_b\,:\,y\in Y,\: b\in B\bigr\}
\end{equation}
are collections of positive semidefinite operators on $\H$, representing
measurements for each choice of $x\in X$ and $y\in Y$.
It is required that each of Alice's measurements commutes with each of Bob's
measurements, meaning that
\begin{equation}
  \bigl[ A^x_a,B^y_b\bigr] = 0
\end{equation}
for all $x\in X$, $y\in Y$, $a\in A$, and $b\in B$.
Similar to standard quantum strategies, there is no generality lost in
considering only projective measurements for Alice and Bob.

The expected pay-off for a commuting measurement strategy, as just described,
in an extended nonlocal game $G = (\pi,V)$, is given by
\begin{equation}
  \label{eq:extended-nlg-commuting-prob}
  \sum_{(x,y) \in X\times Y} \pi(x,y) \sum_{(a,b) \in A\times B}
  \bigip{V(a,b|x,y) \otimes A_a^x B_b^y}{\rho}.
\end{equation}
The \emph{commuting measurement value} of $G$, which is denoted $\omega_c(G)$,
is the supremum value of the expected pay-off of $G$ taken over all commuting
measurement strategies for Alice and Bob.

Along similar lines to standard quantum strategies, a commuting measurement
strategy as above defines a function
$K:A\times B\times X\times Y\rightarrow\Pos(\R)$ as
\begin{equation}
  K(a,b|x,y) = 
  \tr_{\H} \Bigl( \bigl(\I_{\R} \otimes A_a^x B_b^y\bigr) \rho \Bigr)
\end{equation}
for each $x\in X$, $y \in Y$, $a\in A$, and $b\in B$.

  Any function arising from a commuting measurement strategy in this way will be
  called a \emph{commuting measurement assemblage}.
  It is unknown if every commuting measurement assemblage $K$ is induced by a
  standard quantum strategy---as shown by Fritz \cite{Fritz2012}, this problem
  is closely related to the long-standing Connes Embedding conjecture.


\section{The NPA hierarchy for extended nonlocal games}
\label{sec:upper-bound-the-npa-referee-hierarchy}

In this section we describe how the semidefinite programming hierarchy of
Navascu\'{e}s, Pironio, and Ac\'{i}n~\cite{Navascues2007, Navascues2008}
may be generalized to extended nonlocal games.
We will begin by describing the construction of the hierarchy, and then prove
that the hierarchy converges to the commuting measurement value of an extended
nonlocal game.

\subsubsection*{Construction of the extended NPA hierarchy}

Assume that finite and nonempty question and answer sets $X$, $Y$, $A$, and
$B$, as well as a positive integer $m$ representing the dimension of the
referee's quantum system, have been fixed.
We first introduce three alphabets:
\begin{equation}
  \Sigma_A = X \times A,\quad
  \Sigma_B = Y \times B,\quad\text{and}\quad
  \Sigma = \Sigma_A \cupdot \Sigma_B.
\end{equation}
Here, $\cupdot$ denotes the disjoint union, meaning that $\Sigma_A$ and
$\Sigma_B$ are to be treated as disjoint sets when forming $\Sigma$.
For every nonnegative integer $k$, we will write $\Sigma^{\leq k}$ to denote
the set of strings over the alphabet $\Sigma$ having length at most $k$, we
write $\Sigma^{\ast}$ to denote the set of all strings (of finite length) over
$\Sigma$, and we write $\varepsilon$ to denote the empty string.

Next, define $\sim$ to be the equivalence relation on $\Sigma^{\ast}$
generated by the following rules:\vspace{2mm}

\noindent
\begin{tabular}{@{\hspace*{1.2mm}}lll}
  1. \hspace*{-3mm} & $s \sigma t \sim s \sigma \sigma t$ & (for every
  $s,t\in\Sigma^{\ast}$ and $\sigma \in \Sigma$).\\[1mm]
  2. \hspace*{-3mm} & $s \sigma \tau t \sim s \tau \sigma t$ & 
  (for every $s,t\in\Sigma^{\ast}$, $\sigma \in \Sigma_A$, and 
  $\tau\in\Sigma_B$).\\[2mm]
\end{tabular}

\noindent
That is, two strings are equivalent with respect to the relation $\sim$
if and only if one can be obtained from the other by a finite number of
applications of the above rules.

Now, a function of the form 
\begin{equation}
  \phi: \Sigma^{\ast} \rightarrow \complex
\end{equation}
will be said to be \emph{admissible} if and only if the following conditions are
satisfied:
\begin{mylist}{\parindent}
\item[1.] For every choice of strings $s,t\in\Sigma^{\ast}$ it holds that
  \begin{equation}
    \sum_{a\in A} \phi(s (x,a) t) = \phi(st)
    \quad\text{and}\quad
    \sum_{b\in B} \phi(s (y,b) t) = \phi(st)
  \end{equation}
  for every $x\in X$ and $y\in Y$.
\item[2.]
  For every choice of strings $s,t\in \Sigma^{\ast}$, it holds that
  \begin{equation}
    \phi(s (x,a) (x,a') t) = 0 \quad\text{and}\quad
    \phi(s (y,b) (y,b') t) = 0
  \end{equation}
  for every choice of $x\in X$ and $a,a'\in A$ satisfying $a \not= a'$, 
  and every choice of $y\in Y$ and $b,b'\in B$ satisfying $b \not= b'$, 
  respectively.
\item[3.]
  For all strings $s,t\in\Sigma^{\ast}$ satisfying $s\sim t$ it
  holds that $\phi(s) = \phi(t)$.
\end{mylist}
Along similar lines, a function of the form
\begin{equation}
  \phi: \Sigma^{\leq k} \rightarrow \complex
\end{equation}
is said to be \emph{admissible} if and only if the same conditions listed above
hold, provided that $s$ and $t$ are sufficiently short so that $\phi$ is defined
on the arguments indicated within each condition.

Finally, for each positive integer $k$ (representing a level of approximation in
the hierarchy to be constructed), we consider the set of all block matrices of
the form
\begin{equation}
  \label{eq:block-matrix}
  M^{(k)}
  = \begin{pmatrix}
    M^{(k)}_{1,1} & \cdots & M^{(k)}_{1,m}\\
    \vdots & \ddots & \vdots\\
    M^{(k)}_{m,1} & \cdots & M^{(k)}_{m,m}
  \end{pmatrix},
\end{equation}
where each of the blocks takes the form
\begin{equation}
  M^{(k)}_{i,j} : \Sigma^{\leq k} \times \Sigma^{\leq k} \rightarrow \complex,
\end{equation}
and for which the following conditions are satisfied:
\begin{mylist}{\parindent}
\item[1.]
  For every choice of $i,j\in\{1,\ldots,m\}$, there exists an admissible
  function
  \begin{equation}
    \phi_{i,j}:\Sigma^{\leq 2k}\rightarrow \complex
  \end{equation}
  such that
  \begin{equation}
    M^{(k)}_{i,j}(s,t) = \phi_{i,j}(s^{\mathsmaller{R}}t)
  \end{equation}
  for every choice of strings $s,t\in\Sigma^{\leq k}$.
  (Here, the notation $s^{\mathsmaller{R}}$ means the \emph{reverse} of the
  string $s$.)
\item[2.]
  It holds that
  \begin{equation}
    \label{eq:pseudo-correlation-constraint-3}
    M^{(k)}_{1,1}(\varepsilon,\varepsilon) + \cdots +
    M^{(k)}_{m,m}(\varepsilon,\varepsilon) = 1.
  \end{equation}
\item[3.]
  The matrix $M^{(k)}$ is positive semidefinite.
\end{mylist}
Matrices of the form \eqref{eq:block-matrix} obeying the listed constraints will
be called \emph{$k$-th order admissible matrices}.
For such a matrix, we write $M^{(k)}(s,t)$ to denote the $m\times m$ complex
matrix
\begin{equation}
  M^{(k)}(s,t)
  = \begin{pmatrix}
    M^{(k)}_{1,1}(s,t) & \cdots & M^{(k)}_{1,m}(s,t)\\
    \vdots & \ddots & \vdots\\
    M^{(k)}_{m,1}(s,t) & \cdots & M^{(k)}_{m,m}(s,t)
  \end{pmatrix},
\end{equation}
for each choice of strings $s,t\in\Sigma^{\leq k}$.
With respect to this notation, the second and third conditions on $M^{(k)}$
imply that $M^{(k)}(\varepsilon,\varepsilon)$ is an $m\times m$ density matrix.

We observe that an optimization over all $k$-th order admissible
matrices can be represented by a semidefinite program: a matrix of the form
\eqref{eq:block-matrix} is a $k$-th order admissible matrix if and only if it
is positive semidefinite and satisfies a finite number of linear constraints
imposed by the first two conditions on $M^{(k)}$.
In particular, for an extended nonlocal game $G = (\pi,V)$,
where $\pi$ is a distribution over $X\times Y$ and $V$ is a function
$V:A \times B \times X\times Y\rightarrow\Herm(\complex^m)$, 
one may consider the maximization of the quantity
\begin{equation}
  \sum_{x,y,a,b} \pi(x,y) \Bigip{V(a,b|x,y)}{M^{(k)}((x,a),(y,b))}
\end{equation}
subject to $M^{(k)}$ being a $k$-th order admissible matrix.

We also note that the hierarchy of Navascu\'{e}s, Pironio, and Ac\'{i}n
corresponds precisely to the $m=1$ case of the hierarchy just described.

\subsubsection*{Convergence of the extended NPA hierarchy}

Now, for a fixed choice of sets $X$, $Y$, $A$, and $B$, as well as positive
integers $m$ and $k$, let us consider the set of all functions of the form
\begin{equation}
  K : A\times B \times X\times Y \rightarrow \Lin(\complex^m)
\end{equation}
for which there exists a $k$-th order admissible matrix $M^{(k)}$ that satisfies
\begin{equation}
  K(a,b|x,y) = M^{(k)}((x,a),(y,b))
\end{equation}
for every $x\in X$, $y\in Y$, $a\in A$, and $b\in B$.

  The set of all such functions will be called
  \emph{$k$-th order pseudo commuting measurement assemblages}.

\begin{theorem}
  Let $X$, $Y$, $A$, and $B$ be finite sets, let $m$ be a positive integer, and
  let
  \begin{equation}
    K : A\times B\times X\times Y \rightarrow \Lin(\complex^m)
  \end{equation}
  be a function.
  The following statements are equivalent:
  \begin{mylist}{\parindent}
  \item[1.]
    The function $K$ is a commuting measurement assemblage.
  \item[2.]
    The function $K$ is a $k$-th order pseudo commuting measurement assemblage for every
    positive integer $k$.
  \end{mylist}
\end{theorem}

\begin{proof}
  The simpler implication is that statement 1 implies statement 2.
  Under the assumption that statement 1 holds, it must be that $K$ is defined
  by a strategy in which Alice and Bob use projective measurements,
  $\{A_a^x : a \in A\}$ for Alice and $\{B_b^y : b \in B \}$ for Bob,
  on a shared Hilbert space $\H$, along with a pure state $u\in\R\otimes\H$.
  Let $u_1,\ldots,u_m \in \H$ be vectors for which
  \begin{equation}
    u = \sum_{j = 1}^m e_j \otimes u_j.
  \end{equation}
  Also let $\Pi^z_c$ denote $A^z_c$ if $z\in X$ and $c\in A$, or $B^z_c$ if
  $z\in Y$ and $c\in B$.
  With respect to this notation, one may consider the $k$-th order
  admissible matrix $M^{(k)}$ defined by
  \begin{equation}
    M^{(k)}_{i,j}(s,t) = \phi_{i,j}(s^{\mathsmaller{R}}t),
  \end{equation}
  where the functions $\{\phi_{i,j}\}$ are defined as
  \begin{equation}
    \phi_{i,j} \bigl((z_1, c_1) \cdots (z_\ell, c_\ell)\bigr) 
    = u_i^* \Pi_{c_1}^{z_1} \cdots \Pi_{c_\ell}^{z_\ell} u_j
  \end{equation}
  for every string $(z_1, c_1) \cdots (z_\ell, c_\ell)\in\Sigma^{\leq 2k}$.
  A verification reveals that this matrix is consistent with $K$, and therefore
  $K$ is a $k$-th order pseudo commuting measurement
    assemblage.

  The more difficult implication is that statement 2 implies statement 1.
  The basic methodology of the proof is similar to the $m=1$ case proved in
  \cite{Navascues2008}, and we will refer to arguments made in that paper
  when they extend to the general case.
  For every positive integer $k$, let $M^{(k)}$ be a $k$-th order admissible
  matrix satisfying $K(a,b|x,y) = M^{(k)}((x,a),(y,b))$ for every $x\in X$,
  $y\in Y$, $a\in A$, and $b\in B$.

  First, one may observe that for every choice of $k\geq 1$, it holds that
  \begin{equation}
    \Bigabs{M_{i,j}^{(k)}(s,t)} \leq 1
  \end{equation}
  for every choice of $i,j\in\{1,\ldots,m\}$ and $s,t\in\Sigma^{\leq k}$.
  To see that this is so, observe first that
  \begin{equation}
    \Bigabs{M_{i,j}^{(k)}(s,t)} \leq 
    \sqrt{M_{i,i}^{(k)}(s,s)\vphantom{M_{j,j}^{(k)}(t,t)}}
    \sqrt{M_{j,j}^{(k)}(t,t)}
  \end{equation}
  for each $i,j\in\{1,\ldots,m\}$ and $s,t\in\Sigma^{\ast}$, which is a
  consequence of the fact that each $2\times 2$ submatrix
  \begin{equation}
    \begin{pmatrix}
      M_{i,i}^{(k)}(s,s) & M_{i,j}^{(k)}(s,t) \\[2mm]
      M_{j,i}^{(k)}(t,s) & M_{j,j}^{(k)}(t,t)
    \end{pmatrix}
  \end{equation}
  of $M^{(k)}$ is positive semidefinite.
  It therefore suffices to prove that 
  \begin{equation}
    \label{eq:diagonal-bound}
    M_{i,i}^{(k)}(s,s) \leq 1
  \end{equation}
  for every $i\in\{1,\ldots,m\}$ and $s\in\Sigma^{\leq k}$.
  The bound \eqref{eq:diagonal-bound} may be proved by induction on the length
  of~$s$.
  For the base case, one has that
  $M_{i,i}^{(k)}(\varepsilon,\varepsilon) \leq 1$
  by the constraint \eqref{eq:pseudo-correlation-constraint-3}, along with the
  fact that the diagonal entries of $M^{(k)}$ are nonnegative.
  For the general case, one has that for any string $t\in\Sigma^{\ast}$ and any
  choice of $(z,c)\in\Sigma$, it holds that
  \begin{equation}
    \begin{multlined}
      M_{i,i}^{(k)}((z,c)t,(z,c)t)
      \leq
      \sum_d M_{i,i}^{(k)}((z,d)t,(z,d)t)
      = \sum_d
      \phi_{i,i}^{(k)}(t^{\mathsmaller{R}}(z,d)(z,d)t)\\
      = \sum_d
      \phi_{i,i}^{(k)}(t^{\mathsmaller{R}}(z,d)t)
      = \phi_{i,i}^{(k)}(t^{\mathsmaller{R}}t)
      = M_{i,i}^{(k)}(t,t),
    \end{multlined}
  \end{equation}
  where the sums are over all $d\in A$ or $d\in B$ depending on whether
  $z\in X$ or $z\in Y$, respectively.
  By the hypothesis of induction the required bound \eqref{eq:diagonal-bound}
  follows.

  Next, reasoning in the same way as \cite{Navascues2008} through the use of
  the Banach--Alaoglu theorem, one finds that there must exist an infinite
  matrix of the form
  \begin{equation}
    M
    = \begin{pmatrix}
      M_{1,1} & \cdots & M_{1,m}\\
      \vdots & \ddots & \vdots\\
      M_{m,1} & \cdots & M_{m,m}
    \end{pmatrix},
  \end{equation}
  where
  \begin{equation}
    M_{i,j} : \Sigma^{\ast} \times \Sigma^{\ast} \rightarrow \complex
  \end{equation}
  for each $i,j\in\{1,\ldots,m\}$, satisfying similar constraints to the finite
  matrices $M^{(k)}$.
  In particular, it must hold that 
  \begin{equation}
    M_{i,j}(s,t) = \phi_{i,j}(s^{\mathsmaller{R}}t)
  \end{equation}
  for a collection of admissible functions $\{\phi_{i,j}\}$ taking the form
  \begin{equation}
    \phi_{i,j}:\Sigma^{\ast} \rightarrow \complex,
  \end{equation}
  it must hold that all finite submatrices of $M$ are positive semidefinite,
  and it must hold that
  $M_{1,1}(\varepsilon,\varepsilon) + \cdots +
  M_{m,m}(\varepsilon,\varepsilon) = 1$.
  Consequently, there must exist a collection of vectors
  \begin{equation}
    \label{eq:vectors-in-H}
    \bigl\{u_{i,s}\,:\,i\in\{1,\ldots,m\},\;s\in\Sigma^{\ast}\}\subset\H
  \end{equation}
  chosen from a (separable) Hilbert space $\H$ for which it holds that
  \begin{equation}
    M_{i,j}(s,t) = \bigip{u_{i,s}}{u_{j,t}}
  \end{equation}
  for every choice of $i,j\in\{1,\ldots,m\}$ and $s,t\in\Sigma^{\ast}$.
  Furthermore, it must hold that
  \begin{equation}
    K(a,b|x,y) = M((x,a),(y,b))
  \end{equation}
  where, as for the matrices $M^{(k)}$, we write
  \begin{equation}
    M(s,t) = 
    \begin{pmatrix}
      M_{1,1}(s,t) & \cdots & M_{1,m}(s,t)\\
      \vdots & \ddots & \vdots\\
      M_{m,1}(s,t) & \cdots & M_{m,m}(s,t)
    \end{pmatrix}
  \end{equation}
  for each $s,t\in\Sigma^{\ast}$.
  There is no loss of generality in assuming $\H$ is spanned by
  the vectors \eqref{eq:vectors-in-H}, for otherwise $\H$ can simply be replaced
  by the (possibly finite-dimensional) subspace spanned by these vectors.

  Now we will define a commuting measurement strategy for Alice and Bob
  certifying that $K$ is a commuting measurement assemblage.
  The state initially prepared by Alice an Bob, and shared with the referee,
  will be the pure state corresponding to the vector
  \begin{equation}
    u = \sum_{j = 1}^m e_j \otimes u_{j,\varepsilon} \in \complex^m\otimes\H.
  \end{equation}
  This is a unit vector, as a calculation reveals:
  \begin{equation}
    \norm{u}^2 = \sum_{j = 1}^m 
    \ip{u_{j,\varepsilon}}{u_{j,\varepsilon}}
    = M_{1,1}(\varepsilon,\varepsilon) + \cdots +
    M_{m,m}(\varepsilon,\varepsilon) = 1.
  \end{equation}

  Next we define projective measurements on $\H$ for Alice and Bob.
  For each $(z,c)\in\Sigma$, define $\Pi^z_c$ to be the projection operator
  onto the span of the set
  \begin{equation}
    \bigl\{u_{j,(z,c)s}\,:\,j\in\{1,\ldots,m\},\;s\in\Sigma^{\ast}\bigr\}.
  \end{equation}
  It must, of course, be proved that these projections do indeed form projective
  measurements, and that Alice's measurements commute with Bob's.
  Toward these goals, consider any choice of $i,j\in\{1,\ldots,m\}$,
  $s,t\in\Sigma^{\ast}$, and $(z,c)\in\Sigma$, and observe that
  \begin{equation}
    \begin{multlined}
      \bigip{u_{i,(z,c)t}}{u_{j,s}}
      = M_{i,j}((z,c)t,s) = \phi_{i,j}(t^{\mathsmaller{R}}(z,c)s)\\
      = \phi_{i,j}(t^{\mathsmaller{R}}(z,c)(z,c)s)
      = M_{i,j}((z,c)t,(z,c)s) 
      = \bigip{u_{i,(z,c)t}}{u_{j,(z,c)s}}.
    \end{multlined}
  \end{equation}
  It follows that $u_{j,s}$ and $u_{j,(z,c)s}$ have the same inner product
  with every vector in the image of $\Pi^z_c$.
  As every vector in the orthogonal complement of the image of $\Pi^z_c$ is
  obviously orthogonal to $u_{j,(z,c)s}$, as this vector is contained in the
  image of $\Pi^z_c$, it follows that
  \begin{equation}
    \label{eq:projection-on-vectors}
    \Pi^z_c u_{j,s} = u_{j,(z,c)s}.
  \end{equation}
  This formula greatly simplifies the required verifications.
  For instance, one has
  \begin{equation}
    \bigip{u_{i,(z,c)t}}{u_{j,(z,d)s}}
    = M_{i,j}((z,c)t,(z,d)s) = \phi_{i,j}(t^{\mathsmaller{R}}(z,c)(z,d)s)
    = 0
  \end{equation}
  for all $i,j\in\{1,\ldots,m\}$, $s,t\in\Sigma^{\ast}$, and
  $(z,c),(z,d)\in\Sigma$ for which $c\not=d$, and therefore
  $\Pi^z_c \Pi^z_d = 0$
  whenever $(z,c),(z,d)\in\Sigma$ satisfy $c\not=d$.
  For each $x\in X$, and each $i,j\in\{1,\ldots,m\}$ and
  $s,t\in\Sigma^{\ast}$, it holds that
  \begin{equation}
    \sum_{a\in A} \bigip{u_{i,s}}{\Pi^x_a u_{j,t}}
    = \sum_{a\in A} \bigip{u_{i,s}}{u_{j,(x,a)t}}
    = \sum_{a\in A} \phi_{i,j}(s^{\mathsmaller{R}}(x,a)t)
    = \phi_{i,j}(s^{\mathsmaller{R}}t) = \bigip{u_{i,s}}{u_{j,t}}
  \end{equation}
  and therefore
  \begin{equation}
    \sum_{a\in A}\Pi^x_a = \I, 
  \end{equation}
  for each $x\in X$, and along similar lines one finds that
  \begin{equation}
    \sum_{b\in B}\Pi^y_b = \I
  \end{equation}
  for each $y\in Y$.
  Finally, for every $i,j\in\{1,\ldots,m\}$, $s,t\in\Sigma^{\ast}$, 
  $(x,a)\in\Sigma_A$, and $(y,b)\in\Sigma_B$ we have
  \begin{equation}
    \begin{multlined}
      \Bigip{u_{i,s}}{\Pi^x_a \Pi^y_b u_{j,t}}
      = \bigip{u_{i,(x,a)s}}{u_{j,(y,b)t}}
      = \phi_{i,j}\bigl(s^{\mathsmaller{R}}(x,a)(y,b)t\bigr)\\
      = \phi_{i,j}\bigl(s^{\mathsmaller{R}}(y,b)(x,a)t\bigr)
      = \bigip{u_{i,(y,b)s}}{u_{j,(x,a)t}}
      = \bigip{u_{i,s}}{\Pi^y_b \Pi^x_a u_{j,t}},
    \end{multlined}
  \end{equation}
  and therefore $\bigl[\Pi^x_a,\Pi^y_b\bigr] = 0$.

  It remains to observe that the strategy represented by the pure state
  $u$ and the projective measurements $\{\Pi^x_a\}$ and $\{\Pi^y_b\}$
  yields the commuting measurement assemblage $K$.
  This is also evident from the equation \eqref{eq:projection-on-vectors}, as
  one has
  \begin{equation}
    M_{i,j}((x,a),(y,b)) =
    \bigip{u_{i,(x,a)}}{u_{j,(y,b)}} =
    \Bigip{\Pi^x_a \Pi^y_b}{u_{j,\varepsilon} u_{i,\varepsilon}^{\ast}},
  \end{equation}
  and therefore
  \begin{equation}
    K(a,b|x,y) = \tr_{\H} \Bigl( \bigl(\I\otimes \Pi^x_a \Pi^y_b\bigr) u
    u^{\ast}\Bigr)
  \end{equation}
  for every choice of $x\in X$, $y\in Y$, $a\in A$, and $b\in B$.
\end{proof}


\section{Monogamy-of-entanglement games}
\label{sec:monogamy-of-entanglement-games}

As suggested in the introduction, a monogamy-of-entanglement game is specified
by a pair $G = (\pi,R)$ where $\pi:X\rightarrow[0,1]$ is a probability vector
defined over a finite, nonempty set $X$ and $R$ is a function of the form 
$R:A\times X \rightarrow\Pos(\complex^m)$ satisfying 
\begin{equation}
  \sum_{a\in A} R(a|x) = \I
\end{equation}
for every $x\in X$, where $A$ is a finite and nonempty set.
The function $R$ specifies a collection of measurements, one for each choice of
$x\in X$, each having outcomes in $A$.

Recall that in a monogamy-of-entanglement game, Alice and Bob prepare a state,
and then share it with the referee. 
The referee randomly selects a single question $x\in X$, performs a 
measurement $\{R(a|x)\,:\,a\in A\}$ on its portion of the shared state,  
and then sends $x$ to both Alice and Bob. 
The game is won if and only if the responses that Alice and Bob give agree with
the outcome of the referee's measurement.

Because Alice and Bob only win when their output is the same, the optimal
winning probability for an entangled strategy making use of a specific choice
of measurements $\{A^x_a\}$ and $\{B^y_b\}$ for Alice and Bob is given by
\begin{equation}
  \Biggnorm{\sum_{x\in X} \pi(x) \sum_{a\in A} 
    R(a|x) \otimes A^x_a \otimes B^x_a}.
\end{equation}
The unentangled value of a monogamy-of-entanglement game may be expressed as
\begin{equation} \label{unentangled_value}
  \omega(G) = \max_{f:X\rightarrow A}
  \Bignorm{\sum_{x\in X} \pi(x) R(f(x)|x)}.
\end{equation}

As an example of a monogamy-of-entanglement game, we consider the BB84 monogamy game, which was
also introduced in~\cite{Tomamichel2013}.

\begin{example}[BB84 monogamy game]
  Let $m = 2$, let $X = A = \{0,1\}$, and define
  \begin{equation}
    R(0|0) = \ket{0}\bra{0}, \quad
    R(1|0) = \ket{1}\bra{1}, \quad
    R(0|1) = \ket{+}\bra{+}, \quad\text{and}\quad
    R(1|1) = \ket{-}\bra{-}.
  \end{equation}
  Also define $\pi(0) = \pi(1) = 1/2$, and define
  the BB84 monogamy-of-entanglement game $G_{\BB84} = (\pi,R)$.
  It was observed in \cite{Tomamichel2013} that
  \begin{equation}
    \omega(G_{\BB84}) = \omega^{\ast}(G_{\BB84}) = \cos^2(\pi/8).
  \end{equation}
\end{example}

\subsubsection*{Entangled versus unentangled strategies for
  monogamy-of-entanglement games}

The phenomenon that entanglement does not help in the BB84
monogamy-of-entanglement game is not limited to that game.
We show that for any monogamy-of-entanglement game $G$ for which $\abs{X} = 2$,
it must hold that $\omega(G) = \omega^{\ast}(G)$.

\begin{theorem} 
  \label{thm:classical-equal-quantum}
  Let $G$ be any monogamy-of-entanglement game for which it holds that the
  question set $X$ satisfies $\abs{X} = 2$.
  It holds that
  \begin{equation}
    \omega(G) = \omega^{\ast}(G).
  \end{equation}
\end{theorem}

\begin{proof}
  It is evident that $\omega(G) \leq \omega^{\ast}(G)$, as this is so for every
  monogamy-of-entanglement game, so it remains to prove the reverse inequality.

  Assume without loss of generality that $X = \{0,1\}$, assume that
  $G = (\pi,R)$ for $\pi(0) = \lambda$ and $\pi(1) = 1 - \lambda$.
  Consider any choice of projective measurements
  \begin{equation}
    \label{eq:Alice's-measurements}
    \bigl\{A^0_a\,:a\in A\bigr\} \quad\text{and}\quad
    \bigl\{A^1_a\,:a\in A\bigr\}
  \end{equation}
  on $\A$ for Alice and 
  \begin{equation}
    \label{eq:Bob's-measurements}
    \bigl\{B^0_a\,:a\in A\bigr\} \quad\text{and}\quad
    \bigl\{B^1_a\,:a\in A\bigr\}
  \end{equation}
  on $\B$ for Bob.
  The winning probability for a strategy using these measurements is given by
  \begin{equation}
    \label{eq:value-achieved}
    \Biggnorm{\lambda \sum_{a\in A}
      R(a|0) \otimes A^0_a \otimes B^0_a
      + (1-\lambda) \sum_{a\in A} R(a|1) \otimes A^1_a \otimes B^1_a}
  \end{equation}
  for an optimal choice of the initial state.
  For any choice of positive semidefinite operators $P \leq Q$ it holds that
  $\norm{P} \leq \norm{Q}$, from which it follows that
  \eqref{eq:value-achieved} is upper-bounded by
  \begin{equation}
    \label{eq:entangled-monogamy-upper-bound}
    \begin{aligned}
      \Biggnorm{\lambda \sum_{a\in A}
        R(a|0) \otimes A^0_a \otimes \I
        + (1-\lambda) \sum_{b\in A} R(b|1) \otimes \I \otimes B^1_b}
      \hspace{-6cm}\\
      & = \Biggnorm{\sum_{a\in A}\sum_{b\in B}
        \bigl( \lambda R(a|0) + (1 - \lambda) R(b|1)\bigr)
        \otimes A^0_a \otimes B^1_b}\\
      & = \max_{a,b\in A} 
      \bignorm{\lambda R(a|0) + (1 - \lambda) R(b|1)}.
    \end{aligned}
  \end{equation}
  The second equality follows from the fact that
  $\{A^0_a\otimes B^1_b\,:\, a,b\in A\}$ is a collection of pairwise orthogonal
  projection operators.
  The final expression of \eqref{eq:entangled-monogamy-upper-bound} is
  equal to the unentangled value $\omega(G)$ of $G$.
  Because the projective measurements \eqref{eq:Alice's-measurements} and
  \eqref{eq:Bob's-measurements} were chosen arbitrarily, and every entangled
  strategy is equivalent to one in which Alice and Bob use projective
  measurements, it follows that $\omega^{\ast}(G) \leq \omega(G)$ as required.
\end{proof}

An operational interpretation of this result was suggested to us by Thomas
Vidick. 
To convert a quantum strategy into a classical strategy, we can assign
one question to each player (say $0$ to Alice and $1$ to Bob). 
Even before the referee asks the question, Alice can measure her part of the state
with $\{A^0_a\}$ and Bob with $\{B^1_a\}$. 
They exchange their answers and then are separated. 
If the referee asks the question $0$ (or $1$) then they answer according to
Alice (respectively Bob).

It turns out that monogamy-of-entanglement games for which there are more than
two questions can exhibit an advantage of entangled over unentangled
strategies.
The following example describes such a game.

\begin{example}
  Let $\zeta = e^{\frac{2 \pi i}{3}}$ and consider the following four
  mutually unbiased bases:
  \begin{equation}\label{eq:MUB43}
    \begin{aligned}
      \B_0 &= \left\{ \ket{0},\: \ket{1},\: \ket{2} \right\}, \\
      \B_1 &= \left\{ \frac{\ket{0} + \ket{1} + \ket{2}}{\sqrt{3}},\:
      \frac{\ket{0} + \zeta^2 \ket{1} + \zeta \ket{2}}{\sqrt{3}},\:
      \frac{\ket{0} + \zeta\ket{1} + \zeta^2 \ket{2}}{\sqrt{3}} \right\}, \\
      \B_2 &= \left\{ \frac{\ket{0} + \ket{1} + \zeta\ket{2}}{\sqrt{3}},\:
      \frac{\ket{0} + \zeta^2 \ket{1} + \zeta^2 \ket{2}}{\sqrt{3}},\:
      \frac{\ket{0} + \zeta \ket{1} + \ket{2}}{\sqrt{3}} \right\}, \\
      \B_3 &= \left\{ \frac{\ket{0} + \ket{1} + \zeta^2 \ket{2}}{\sqrt{3}},\:
      \frac{\ket{0} + \zeta^2 \ket{1} + \ket{2}}{\sqrt{3}},\:
      \frac{\ket{0} + \zeta \ket{1} + \zeta \ket{2}}{\sqrt{3}} \right\}.
    \end{aligned}
  \end{equation}
  Define a monogamy-of-entanglement game $G = (\pi,R)$ so that
  \begin{equation}
    \pi(0) = \pi(1) = \pi(2) = \pi(3) = \frac{1}{4}
  \end{equation}
  and $R$ is such that
  \begin{equation}
    \bigl\{ R(0|x),\:R(1|x),\:R(2|x)\bigr\}
  \end{equation}
  represents a measurement with respect to the basis $\B_x$, for each
  $x\in \{0,1,2,3\}$.
  An exhaustive search over all unentangled strategies reveals that
  \begin{equation}
    \omega(G) = \frac{3+\sqrt{5}}{8} \approx 0.6545,
  \end{equation}
  while a computer search over quantum strategies has revealed that
  \begin{equation}
    \omega^{\ast}(G) \geq 0.660986,
  \end{equation}
  which is strictly larger than the unentangled value of this game.
  (This strategy is available for download from the software
  repository~\cite{Johnston2015b}.
  We do not know if this strategy is optimal---the first level of the extended
  NPA hierarchy of Section~\ref{sec:upper-bound-the-npa-referee-hierarchy}
  gives an upper bound of $2/3$ on the commuting measurement value of this
  game.)
\end{example}

\subsubsection*{Parallel repetition of monogamy-of-entanglement games}
\label{sec:parallel-repetition-of-monogamy-games}

Tomamichel et al.~\cite{Tomamichel2013} proved the following upper bound on the
value of monogamy-of-entanglement games when they are repeated in parallel,
under the assumption that the distribution $\pi$ is uniform over the question
set $X$.
They also proved that this bound is tight for the BB84 monogamy-of-entanglement
game.

\begin{theorem}[Tomamichel, Fehr, Kaniewski, and Wehner]
  \label{thm:upper-parallel-rep}
  Let $G = (\pi,R)$ be a monogamy-of-entanglement game for which $\pi$ is
  uniform over $X$, define
  \begin{equation}
    c(G) = \max_{\substack{x,y \in X \\ x \not= y}} \max_{a,b\in A}
    \biggnorm{\sqrt{R(a|x)} \sqrt{R(b|y)}}^2,
  \end{equation}
  and let $G^{n}$ denote the game $G$ played $n$ times in parallel.
  It holds that
  \begin{equation}
    \omega^{\ast}(G^n) \leq
    \left( \frac{1}{\abs{X}} + \frac{\abs{X} - 1}{\abs{X}}
    \sqrt{c(G)} \right)^n.
  \end{equation}
\end{theorem}

We prove that this bound is, in fact, tight for all monogamy-of-entanglement
games for which $\abs{X} = 2$, the questions are chosen uniformly,
and the referee's measurements are projective.
This is a consequence of the following proposition.

\begin{prop}
  \label{thm:strong-parallel-rep-classical-equal-quantum}
  Let $G = (\pi,R)$ be a monogamy-of-entanglement game for which $X = \{0,1\}$,
  $\pi$ is uniform over $X$, and $R(a|x)$ is a projection operator for each
  $x\in X$ and $a\in A$.
  It holds that
  \begin{equation}
    \omega(G) = 
    \frac{1}{2}+\frac{1}{2} \max_{a,b\in A}\biggnorm{R(a|0)\,R(b|1)}.
  \end{equation}
\end{prop}

\noindent
The proof of this proposition makes use of the following lemma.

\begin{lemma} \label{lem:spectralnorm-projectors} 
  Let $\Pi_0$ and $\Pi_1$ be nonzero projection operators on $\complex^n$.
  It holds that
  \begin{equation}
    \label{lem:spectralnorm-projectors-equation}
    \norm{\Pi_0 + \Pi_1} = 1 + \norm{\Pi_0 \Pi_1}.
  \end{equation}
\end{lemma}

\begin{proof}
  For every choice of unit vectors $u_0, u_1 \in \complex^n$, one has the
  formula
  \begin{equation}
    \bignorm{u_0 u_0^{\ast} + u_1 u_1^{\ast}} =
    1 + \bigabs{\bigip{u_0}{u_1}},
  \end{equation}
  which follows from the observation that the Hermitian operator
  $u_0 u_0^{\ast} + u_1 u_1^{\ast}$ has (at most) two nonzero eigenvalues
  $1 \pm \bigabs{\bigip{u_0}{u_1}}$.
  Letting $\S$, $\S_0$, and $\S_1$ denote the unit spheres in the spaces
  $\complex^n$, $\op{im}(\Pi_0)$, and $\op{im}(\Pi_1)$, respectively, one has
  \begin{equation}
    \begin{aligned}
      \bignorm{\Pi_0 + \Pi_1}
      & = \max\Bigl\{
      v^{\ast}(\Pi_0 + \Pi_1) v\,:\,v\in\S\Bigr\}\\
      & = \max\Bigl\{\norm{\Pi_0 v}^2 + \norm{\Pi_1 v}^2\,:\,v\in\S\Bigr\}\\
      & = \max\Bigl\{\bigabs{\bigip{u_0}{v}}^2 + \bigabs{\bigip{u_1}{v}}^2
      \,:\,v\in\S,\,u_0\in\S_0,\,u_1\in\S_1\Bigr\}\\
      & = \max\Bigl\{ v^{\ast}\bigl(u_0 u_0^{\ast} + u_1 u_1^{\ast}\bigr)v
      \,:\,v\in\S,\,u_0\in\S_0,\,u_1\in\S_1\Bigr\}\\
      & = \max\Bigl\{\bignorm{u_0 u_0^{\ast} + u_1 u_1^{\ast}} \,:\,
      u_0\in\S_0,\,u_1\in\S_1\Bigr\}\\
      & = \max\Bigl\{1 + \bigabs{\bigip{u_0}{u_1}} \,:\,
      u_0\in\S_0,\,u_1\in\S_1\Bigr\}\\
      & = 1 + \bignorm{\Pi_0 \Pi_1},
    \end{aligned}
  \end{equation}
  which proves the lemma.
\end{proof}

\begin{proof}[Proof of
    Proposition~\ref{thm:strong-parallel-rep-classical-equal-quantum}]
  We observe that the unentangled value of $G$ is given by
  \begin{equation}
    \omega(G) = \max_{a,b \in A}\,\biggnorm{\frac{R(a|0) + R(b|1)}{2}}
    = \frac{1}{2}+\frac{1}{2} \max_{a,b\in A}\biggnorm{R(a|0)\,R(b|1)}
  \end{equation}
  as claimed.
\end{proof}

The reason that the proposition just proved implies the tightness of the
bound in Theorem~\ref{thm:upper-parallel-rep} for a monogamy-of-entanglement
game of the type specified in
Proposition~\ref{thm:strong-parallel-rep-classical-equal-quantum} is that
Alice and Bob can simply play, $n$ times in parallel, an optimal strategy for
$G$.
This implies that
\begin{equation}
  \omega^{\ast}(G^n) \geq \omega(G^n)
  \geq 
  \Biggl(
  \frac{1}{2}+\frac{1}{2} \max_{a,b\in A}\biggnorm{R(a|0)\,R(b|1)}\Biggr)^n
  = \left( \frac{1}{2} + \frac{1}{2}\sqrt{c(G)} \right)^n,
\end{equation}
which matches the upper-bound of
Theorem~\ref{thm:upper-parallel-rep}.

%
%
%


\subsection*{Acknowledgments}

We thank Richard Cleve, Debbie Leung, Li Liu, Matt McKague, Jamie Sikora,
Marco Tomamichel, Thomas Vidick, and Elie Wolfe for helpful discussions, and we thank Daniel Cavalcanti, Tobias Fritz, and Marco Piani
  for bringing the connection between extended nonlocal games and multipartite
  quantum steering to our attention.
We also acknowledge Michael Grant and Stephen Boyd for their convex optimization
software CVX~\cite{Grant2008a}. 
RM is supported by the INSPIRE fellowship.
He would like to thank the Institute for Quantum Computing at the University of
Waterloo for its hospitality while contributing to this work.
VR is supported by NSERC and the US Army Research Office. 
JW is supported by NSERC. 

\bibliographystyle{alpha}
\bibliography{refs_jab}

\newcommand{\etalchar}[1]{$^{#1}$}
\begin{thebibliography}{BOGKW88}

\bibitem[BBT05]{Brassard2005}
Gilles Brassard, Anne Broadbent, and Alain Tapp.
\newblock Quantum pseudo-telepathy.
\newblock {\em Foundations of Physics}, 35(11):1877--1907, 2005.

\bibitem[Bel64]{Bell1964}
John Bell.
\newblock On the {E}instein-{P}odolsky-{R}osen paradox.
\newblock {\em Physics}, 1(3):195--200, 1964.

\bibitem[BFL91]{Babai1991}
L{\'a}szl{\'o} Babai, Lance Fortnow, and Carsten Lund.
\newblock Non-deterministic exponential time has two-prover interactive
  protocols.
\newblock {\em Computational {C}omplexity}, 1(1):3--40, 1991.

\bibitem[BFS14]{BuhrmanFS14}
Harry Buhrman, Serge Fehr, and Christian Schaffner.
\newblock On the parallel repetition of multi-player games: The no-signaling
  case.
\newblock In {\em Proceedings of the 9th Conference on the Theory of Quantum
  Computation, Communication and Cryptography}, Leibniz International
  Proceedings in Informatics, pages 24--35. Schloss Dagstuhl, 2014.

\bibitem[BOGKW88]{Ben-Or1988}
Michael Ben-Or, Shafi Goldwasser, Joe Kilian, and Avi Wigderson.
\newblock Multi-prover interactive proofs: How to remove intractability
  assumptions.
\newblock In {\em Proceedings of the {T}wentieth {A}nnual ACM {S}ymposium on
  Theory of {C}omputing}, pages 113--131. ACM, 1988.

\bibitem[CHSH69]{Clauser1969}
John Clauser, Michael Horne, Abner Shimony, and Richard Holt.
\newblock Proposed experiment to test local hidden-variable theories.
\newblock {\em Physical {R}eview {L}etters}, 23(15):880, 1969.

\bibitem[CHTW04]{Cleve2004}
Richard Cleve, Peter Hoyer, Benjamin Toner, and John Watrous.
\newblock Consequences and limits of nonlocal strategies.
\newblock In {\em Computational Complexity, 2004. Proceedings. 19th IEEE Annual
  Conference on}, pages 236--249. IEEE, 2004.

\bibitem[CKW00]{Coffman2000}
Valerie Coffman, Joydip Kundu, and William Wootters.
\newblock Distributed entanglement.
\newblock {\em Physical Review A}, 61(5):052306, 2000.

\bibitem[CM14]{CM14}
Richard Cleve and Rajat Mittal.
\newblock Characterization of binary constraint system games.
\newblock In Javier Esparza, Pierre Fraigniaud, Thore Husfeldt, and Elias
  Koutsoupias, editors, {\em Automata, Languages, and Programming}, volume 8572
  of {\em Lecture Notes in Computer Science}, pages 320--331. Springer Berlin
  Heidelberg, 2014.

\bibitem[CSA{\etalchar{+}}15]{Cavalcanti2015}
Daniel Cavalcanti, Paul Skrzypczyk, Gregory Aguilar, Ranieri Nery, Paulo~Souto
  Ribeiro, and Stephen Walborn.
\newblock Detection of entanglement in asymmetric quantum networks and
  multipartite quantum steering.
\newblock {\em Nature Communications}, 6, 2015.

\bibitem[CSUU08]{Cleve2008}
Richard Cleve, William Slofstra, Falk Unger, and Sarvagya Upadhyay.
\newblock Perfect parallel repetition theorem for quantum {XOR} proof systems.
\newblock {\em Computational Complexity}, 17(2):282--299, 2008.

\bibitem[DLTW08]{Doherty2008}
Andrew Doherty, Yeong-Cherng Liang, Ben Toner, and Stephanie Wehner.
\newblock The quantum moment problem and bounds on entangled multi-prover
  games.
\newblock In {\em Computational Complexity, 2008. CCC'08. 23rd Annual IEEE
  Conference on}, pages 199--210. IEEE, 2008.

\bibitem[DSV13]{Dinur2013}
Irit Dinur, David Steurer, and Thomas Vidick.
\newblock A parallel repetition theorem for entangled projection games.
\newblock {\em Computational Complexity}, 24:201--254, 2013.

\bibitem[Fei91]{Feige1991}
Uriel Feige.
\newblock On the success probability of the two provers in one-round proof
  systems.
\newblock In {\em Structure in Complexity Theory Conference, 1991., Proceedings
  of the Sixth Annual}, pages 116--123. IEEE, 1991.

\bibitem[FK94]{Feige1994}
Uri Feige and Joe Kilian.
\newblock Two prover protocols: {L}ow error at affordable rates.
\newblock In {\em Proceedings of the Twenty-sixth Annual ACM Symposium on
  Theory of Computing}, pages 172--183. ACM, 1994.

\bibitem[For89]{Fortnow1989}
Lance Fortnow.
\newblock {\em Complexity-theoretic aspects of interactive proof systems}.
\newblock PhD thesis, Massachusetts Institute of Technology, 1989.

\bibitem[Fri12]{Fritz2012}
Tobias Fritz.
\newblock Tsirelson's problem and {K}irchberg's conjecture.
\newblock {\em Reviews in Mathematical Physics}, 24(05):1250012, 2012.

\bibitem[GBY08]{Grant2008a}
Michael Grant, Stephen Boyd, and Yinyu Ye.
\newblock {CVX}: {MATLAB} software for disciplined convex programming, 2008.

\bibitem[JP11]{Junge2011a}
Marius Junge and Carlos Palazuelos.
\newblock Large violation of {B}ell inequalities with low entanglement.
\newblock {\em Communications in Mathematical Physics}, 306(3):695--746, 2011.

\bibitem[JR15]{Johnston2015b}
Nathaniel Johnston and Vincent Russo.
\newblock Supplementary software for implementing the examples for the extended
  {NPA} hierarchy of semidefinite programs.
\newblock
  \href{https://github.org/vprusso/monogamy-of-entanglement-games}{https://github.org/vprusso/monogamy-of-entanglement-games},
  2015.

\bibitem[KKM{\etalchar{+}}11]{Kempe2011}
Julia Kempe, Hirotada Kobayashi, Keiji Matsumoto, Ben Toner, and Thomas Vidick.
\newblock Entangled games are hard to approximate.
\newblock {\em SIAM Journal on Computing}, 40(3):848--877, 2011.

\bibitem[KR10]{Kempe2010}
Julia Kempe and Oded Regev.
\newblock No strong parallel repetition with entangled and non-signaling
  provers.
\newblock In {\em Computational Complexity (CCC), 2010 IEEE 25th Annual
  Conference on}, pages 7--15. IEEE, 2010.

\bibitem[KRT10]{Kempe2010a}
Julia Kempe, Oded Regev, and Ben Toner.
\newblock Unique games with entangled provers are easy.
\newblock {\em SIAM Journal on Computing}, 39(7):3207--3229, 2010.

\bibitem[KW04]{Koashi2004}
Masato Koashi and Andreas Winter.
\newblock Monogamy of quantum entanglement and other correlations.
\newblock {\em Physical Review A}, 69(2):022309, 2004.

\bibitem[NPA07]{Navascues2007}
Miguel Navascu\'{e}s, Stefano Pironio, and Antonio Ac{\'\i}n.
\newblock Bounding the set of quantum correlations.
\newblock {\em Physical Review Letters}, 98:010401, 2007.

\bibitem[NPA08]{Navascues2008}
Miguel Navascu{\'e}s, Stefano Pironio, and Antonio Ac{\'\i}n.
\newblock A convergent hierarchy of semidefinite programs characterizing the
  set of quantum correlations.
\newblock {\em New Journal of Physics}, 10(7):073013, 2008.

\bibitem[OV06]{Osborne2006}
Tobias Osborne and Frank Verstraete.
\newblock General monogamy inequality for bipartite qubit entanglement.
\newblock {\em Physical Review Letters}, 96(22):220503, 2006.

\bibitem[Raz98]{Raz1998}
Ran Raz.
\newblock A parallel repetition theorem.
\newblock {\em SIAM Journal on Computing}, 27(3):763--803, 1998.

\bibitem[RV13]{Regev2013}
Oded Regev and Thomas Vidick.
\newblock Quantum {XOR} games.
\newblock In {\em Computational Complexity (CCC), 2013 IEEE Conference on},
  pages 144--155. IEEE, 2013.

\bibitem[SBC{\etalchar{+}}15]{Sainz2015}
Ana~Belen Sainz, Nicolas Brunner, Daniel Cavalcanti, Paul Skrzypczyk, and
  Tam{\'a}s V{\'e}rtesi.
\newblock Post-quantum steering.
\newblock {\em arXiv preprint arXiv:1505.01430}, 2015.

\bibitem[Ter01]{Terhal2001}
Barbara Terhal.
\newblock A family of indecomposable positive linear maps based on entangled
  quantum states.
\newblock {\em Linear Algebra and its Applications}, 323(1):61--73, 2001.

\bibitem[Ter04]{Terhal2004}
Barbara Terhal.
\newblock Is entanglement monogamous?
\newblock {\em IBM Journal of Research and Development}, 48(1):71--78, 2004.

\bibitem[TFKW13]{Tomamichel2013}
Marco Tomamichel, Serge Fehr, Jkedrzej Kaniewski, and Stephanie Wehner.
\newblock A monogamy-of-entanglement game with applications to
  device-independent quantum cryptography.
\newblock {\em New Journal of Physics}, 15(10):103002, 2013.

\bibitem[Vid13]{Vidick2013}
Thomas Vidick.
\newblock Three-player entangled {XOR} games are {NP}-hard to approximate.
\newblock In {\em Foundations of Computer Science (FOCS), 2013 IEEE 54th Annual
  Symposium on}, pages 766--775. IEEE, 2013.

\end{thebibliography}

\end{document}